\documentclass[12pt]{article}
\usepackage{amsmath}
\usepackage{amssymb}
\usepackage[a4paper]{geometry}
\usepackage{amsthm}
\usepackage{subcaption}
\usepackage{graphicx}
\usepackage{tikz}
\usepackage{tikz-cd}
\usepackage{mathrsfs}
\usepackage{bm}
\usepackage{slashed}
\usepackage[hidelinks]{hyperref}
\usepackage{listings}

\hypersetup{
    colorlinks = true,
    linkcolor = {black},
    citecolor = {blue},
    anchorcolor = {green},
    citecolor = {red},
    filecolor = {cyan},
    menucolor = {red},
    runcolor = {cyan},
    urlcolor = {magenta}
}
\newcommand{\twidth}{6in}
\textwidth=\twidth
\renewcommand{\O}{{\mathrm{O}}}
\newcommand{\SU}{{\mathrm{SU}}}
\newcommand{\SO}{{\mathrm{SO}}}

\newcommand{\R}{{\mathbb{R}}}

\newcommand{\Z}{{\mathbb{Z}}}

\newcommand{\C}{{\mathbb{C}}}
\renewcommand{\H}{{\mathbb{H}}}

\newcommand{\1}{{\mathbb{1}}}

\newcommand{\bu}{{\bm{1}}}
\newcommand{\bi}{{\bm{i}}}
\newcommand{\bj}{{\bm{j}}}
\newcommand{\bk}{{\bm{k}}}
\newcommand{\beq}{\begin{equation}}
\newcommand{\eeq}{\end{equation}}
\newcommand{\bea}{\begin{eqnarray}}
\newcommand{\eea}{\end{eqnarray}}
\newcommand{\bal}{\begin{align}}
\newcommand{\eal}{\end{align}}
\newcommand{\bml}{\begin{multline}}
\newcommand{\eml}{\end{multline}}

\newcommand{\bdy}{\partial}

\newcommand{\wt}{\widetilde}

\newcommand{\lto}{\longrightarrow}

\newcommand{\im}{{\rm im}\, }

\newcommand{\id}{{\rm Id}}

\newcommand{\diag}{{\rm diag}}

\makeatletter
\newcommand\xleftrightarrow[2][]{%
  \ext@arrow 9999{\longleftrightarrowfill@}{#1}{#2}}
\newcommand\longleftrightarrowfill@{%
  \arrowfill@\leftarrow\relbar\rightarrow}
\makeatother

\newcommand{\bigslant}[2]{{\raisebox{.2em}{$#1$}\left/\raisebox{-.2em}{$#2$}\right.}}

\newtheorem{theorem}{Theorem}

\begin{document}
\renewcommand*{\thefootnote}{\fnsymbol{footnote}}
\begin{titlepage}
\begin{center}
{\Large
Finkelstein--Rubinstein constraints\\ from ADHM data and rational maps \par}


\vspace{10mm}
{\Large Josh Cork\footnote{Email address: \texttt{josh.cork@leicester.ac.uk}}$^{\rm 1}$\ and Derek Harland\footnote{Email address: \texttt{d.g.harland@leeds.ac.uk}}$^{\rm 2}$}\\[10mm]

\noindent {\em ${}^{\rm 1}$ School of Computing and Mathematical Sciences\\
University of Leicester, University Road, Leicester, United Kingdom
}\\
\smallskip
\noindent {\em ${}^{\rm 2}$ School of Mathematics, University of Leeds\\ Woodhouse Lane, Leeds, United Kingdom }\\[10mm]
{\Large \today}
\vspace{15mm}
\begin{abstract}
    We establish simple formulae for computing Finkelstein--Rubinstein signs for Skyrme fields generated in two ways: from instanton ADHM data, and from rational maps. This may be used to compute homotopy classes of general loops in the configuration spaces of skyrmions, and as a result provide a useful tool for a quantum treatment beyond rigid-body quantisation of skyrmions.
\end{abstract}
\end{center}
\end{titlepage}
\renewcommand*{\thefootnote}{\arabic{footnote}}
\setcounter{footnote}{0}
\hypersetup{
    linkcolor = {blue}
}
\numberwithin{equation}{section}
\section{Introduction}
The Skyrme model \cite{skyrme1962nucl} is a nonlinear model of nuclei, where baryons are interpreted as topological soliton solutions to static field equations called skyrmions. The model is recognised as an effective model of QCD \cite{Witten1983baryons}; since extracting properties of nuclei directly from QCD is beyond current computational capabilities, the Skyrme model acts as a simpler and more tractable alternative.

In this model pions are encoded by maps $U:\R^3\to\SU(2)\cong S^3$ which are continuous and satisfy $U\to\id$ as $|{\bm x}|\to\infty$; these are called Skyrme fields. The space $S$ of all Skyrme fields can be identified with the space of based continuous maps from $S^3$ to $S^3$ via the one-point-compactification $S^3\cong\R^3\cup\{\infty\}$. Each Skyrme field has a topological degree $N\in\pi_3(S^3)\cong\Z$, physically identified as the baryon number, and this splits the configuration space into distinct connected components $S_N$, labelled by topological degree.

In order to make contact with real nuclear physics, one must quantise the model. A key step in this process is to constrain the wave function, which is a map $\Psi:\wt{S_N}\to\C$ defined on the universal cover $\wt{S_N}$ of $S_N$. The space $\wt{S_N}$ admits an action of the fundamental group $\pi_1(S_N)\cong\Z_2=\{\pm1\}$, and the wavefunction is required to satisfy
\begin{equation}\label{FR1}
\Psi((-1)\cdot \wt{U})=-\Psi(\wt{U})
\end{equation}
for all points $\wt{U}\in\wt{S_N}$ \cite{FinkelsteinRubinstein1968,Witten1983baryons}.  In practice, quantisation is often carried out on a subset $C\subset S_N$, lifted to $\wt{C}\subset \wt{S_N}$.  In this case, for closed loops $U^t$ in $C$ parametrised by $t\in [0,1]$ and satisfying $U^1=U^0$, one needs to calculate the corresponding element of $\pi_1(S_N)$, denoted $\sigma(U^t)=\pm1$.  The wavefunction along the lift $\wt{U^t}$ of the loop in $\wt{S_N}$ is required to satisfy
\begin{equation}\label{FR2}
\Psi(\wt{U^1})=\sigma(U^t)\Psi(\wt{U^0}).
\end{equation}
Constraints of the form \eqref{FR2} are known as Finkelstein--Rubinstein constraints \cite{FinkelsteinRubinstein1968}.  The sign $\sigma(U^t)\in\pi_1(S_N)$ is called the Finkelstein--Rubinstein sign, or sometimes the Finkelstein--Rubinstein phase. Finkelstein—-Rubinstein constraints play a vital role in quantisation. They ensure that the nucleon has odd spin and Fermi exchange statistics \cite{FinkelsteinRubinstein1968}, and they are used to calculate spin and isospin quantum numbers of nuclei \cite{BattyeMantonSutcliffeWood2009light,krusch2003homotopy,LauManton2014,mankomantonwood2007,MantonWood2006lithium}.

Various methods for calculating Finkelstein--Rubinstein signs are already known, applicable to Skyrme fields of particular forms.  The simplest way to construct a degree $N$ Skyrme field is to take a product of $N$ degree $1$ Skyrme fields.  This is known as the product approximation.  The Finkelstein--Rubinstein sign of a path generated in the product approximation by rotating and permuting the centres of the degree $1$ Skyrme fields can be calculated explicitly using the original methods of Finkelstein and Rubinstein \cite{FinkelsteinRubinstein1968}.  This process can be automated using an algorithm presented in \cite{GillardHarlandKirkElliotMaybeeSpeight2017}.

Unfortunately, energy-minimising skyrmions rarely resemble products of degree $1$ Skyrme fields.  A more sophisticated approximation, which does accurately describe several minimal-energy Skyrme fields, is the rational map approximation \cite{HoughtonMantonSutcliffe1998rational}.  This constructs a Skyrme field as a suspension of a rational map from $S^2$ to $S^2$.  The calculation of Finkelstein--Rubinstein signs within the rational map approximation was addressed in \cite{krusch2003homotopy}.  In particular, a simple method was developed for calculating Finkelstein--Rubinstein signs of paths induced by rotations of the domain and target.  This method has been extended in \cite{krusch2006finkelstein} and used extensively in rigid body quantisation \cite{BattyeMantonSutcliffeWood2009light,LauManton2014,mankomantonwood2007,MantonWood2006lithium}.  The methods of \cite{krusch2003homotopy,krusch2006finkelstein} are not as easily applied to more general paths of rational maps, because they entail finding roots of degree $N$ polynomials.

Several recent papers on the quantisation of skyrmions have made use of families of Skyrme fields which are not induced by rotations \cite{Gudnason:2018aej,Halcrow:2015rvz,Halcrow:2016spb,Rawlinson:2017rcq, SpeightWinyard2023nudged}.  This motivates the search for more sophisticated methods to approximate Skyrme fields and calculate Finkelstein--Rubinstein signs.  The Atiyah--Manton approximation \cite{AtiyahManton1989} is a very powerful method to generate Skyrme fields which accurately describes numerous low-energy skyrmions \cite{corkhalcrow2022adhm}.  It has fruitfully been used for quantising of skyrmions \cite{Halcrow:2022bxw,Leese:1994hb,walet1996quantising}.  This approximation generates Skyrme fields from Yang-Mills instantons, which are in turn easily produced using the ADHM construction \cite{ADHM1978construction,christWeinbergStanton1978general,MantonSutcliffe2004}.

In this paper we present a method to calculate Finkelstein--Rubinstein signs for paths of Skyrme fields induced from ADHM data.  We also present a new way to calculate Finkelstein--Rubinstein signs for paths generated using the rational map approximation.  Both methods are very simple and widely applicable.  They are not restricted to paths generated by rotations, and they do not require factorisation of polynomials.  The new methods are derived in the next two sections.  The following sections illustrate the methods in some examples and present some concluding remarks.

\section{Instantons and ADHM data}
The Atiyah-Manton approximation constructs Skyrme fields from instantons on $\R^4$.  The most powerful way of creating instantons is the ADHM construction \cite{ADHM1978construction,christWeinbergStanton1978general,MantonSutcliffe2004}, which produces all instantons on $\R^4$ using just quaternionic linear algebra.  It begins with a pair $(L,M)$ of quaternionic matrices $L\in\mathrm{Mat}_{1\times N}(\H)$ and $M\in\mathrm{Mat}_{N\times N}(\H)$, with $M$ symmetric.  From these, we build the operator
\begin{align}
\Delta_x=\begin{pmatrix}L\\M-x\id_N\end{pmatrix},
\end{align}
in which $x\in\H$ represents a point in $\R^4$ via
\begin{align}
    x=x_1\bm{i}+x_2\bm{j}+x_3\bm{j}+x_4\bm{1}.
\end{align}
We assume that $\Delta_x^\dagger \Delta_x$ is invertible for all $x$; then the kernel of $\Delta_x$ is spanned by a quaternionic column vector $V_x$ which can be normalised so that $V_x^\dagger V_x=1$.  From this, one can construct the induced connection
\begin{equation}
A_\mu = V_x^\dagger \frac{\bdy}{\bdy x^\mu}V_x
\end{equation}
on the kernel of $\Delta_x^\dagger$.  This connection is a linear combination of $\bm{i},\bm{j},\bm{k}$, which generate the Lie algebra $\mathfrak{su}(2)$.  It solves the self-dual Yang--Mills equation provided that the $N\times N$ matrix $\Delta_x^\dagger \Delta_x$ is real.

It is easy to see show that the connection $A_\mu$ is unchanged by the transformation
\begin{align}\label{gauge-transformation}
    (L,M)\mapsto (LP^{-1},P MP^{-1}),\quad P\in\O(N).
\end{align}
Denote by $D_N$ the set of all pairs $(L,M)$ satisfying condition that $\Delta_x^\dagger\Delta_x$ is real and invertible, and let $I_N=\bigslant{D_N}{\O(N)}$ be the moduli space of ADHM data.  The above-described construction identifies $I_N$ with the moduli space of framed self-dual Yang--Mills instantons on $\R^4$ with charge $N$.

Atiyah--Manton \cite{AtiyahManton1989} proposed an approximate description of skyrmions by taking holonomy of instantons.  More precisely, a Skyrme field $U(\bm{x})=\Omega(\bm{x},\infty)$ is given by solving
\begin{equation}
\frac{\partial}{\partial x_4}\Omega(\bm{x},x_4)=\Omega(\bm{x},x_4)A_4(\bm{x},x_4),\quad \Omega(\bm{x},-\infty)=\bm{1}.
\end{equation}
This will satisfy the boundary condition $U(\bm{x})\to\bm{1}$ as $|\bm{x}|\to\infty$ provided that $V_x$ is chosen to satisfy $V_x\to (\bu,0,\ldots,0)^t$ as $|x|\to\infty$.

Now suppose that $(L^t,M^t)$ describes a loop in $I_N$, parameterised by $t\in[0,1]$.  This satisfies
\begin{align}\label{loop-cgt}
    (L^1,M^1)=(L^0P^{-1},P M^0P^{-1}).
\end{align}
for some $P\in \O(N)$.  The Atiyah--Manton construction associates to this a loop $U^t$ in $S_N$.  The aim of this section is to prove the following simple method to extract the Finkelstein--Rubinstein sign $\sigma(U^t)\in\pi_1(S_N)$ for this loop.
\begin{theorem}\label{thm:ADHM-homotopy}
    Let $(L^t,M^t)\in I_N$ be a loop of ADHM matrices satisfying \eqref{loop-cgt} as above, and let $U^t$ be the corresponding loop in the space of Skyrme fields. The Finkelstein--Rubinstein sign of this loop is given by
    \begin{align}
        \sigma(U^t)=\det P.
    \end{align}
\end{theorem}
\begin{proof}
    It is known by a result of Hurtubise \cite{hurtubise1986instantons} that $\pi_1(I_N)=\Z_2$. From this it follows that the universal cover $\widetilde{I_N}\to I_N$ is two-to-one.  There is another natural two-to-one cover of $I_N$ given by 
\begin{align}
\Z_2\lto\bigslant{D_N}{\SO(N)}\lto I_N.
\end{align}
The associated exact sequence of homotopy groups includes a homomorphism
\begin{align}\label{HE}
\pi_1(I_N)\lto\pi_0(\Z_2)\cong\Z_2.
\end{align}
This homomorphism $\pi_1(I_N)\to\Z_2$ is precisely the map $[(L^t,M^t)]\mapsto\det P$, where $(L^t,M^t)$ is the loop satisfying \eqref{loop-cgt}.  On the other hand, the Atiyah--Manton construction gives rise to a map from $I_N$ to the space $S_N$ of Skyrme configurations, and hence a homomorphism
\begin{equation}\label{AM}
\pi_1(I_N)\lto \pi_1(S_N)\cong\Z_2.
\end{equation}
So now we have two natural homomorphisms $\pi_1(I_N)\to\Z_2$, and we need to show that these two homomorphisms agree.

To do so, consider the following family of ADHM matrices:
\begin{align}
    \begin{aligned}
        L^t&=\begin{pmatrix}
            1&1&\cdots&1
        \end{pmatrix},\\
        M^t&=\diag\left\{
            \cos(\pi t)\bi+\sin(\pi t)\bj,-\cos(\pi t)\bi-\sin(\pi t)\bj,2\bi,\dots,(N-1)\bi\right\}.
    \end{aligned}\label{loop-adhm-swap}
\end{align}
It is straightforward to check that these satisfy the ADHM constraints, namely that $\Delta_x^\dagger \Delta_x$ is real and invertible for all $x$.  They moreover satisfy
\begin{align*}
    (L^1,M^1)=(L^0P^{-1},PM^0P^{-1}),\quad P=\left(\begin{array}{c|c}
        \begin{array}{cc}0&1\\1&0\end{array}&0\\\hline
        0&\id_{N-2}\end{array}\right).
\end{align*}
Since $\det P=-1$, it follows that the homomorphism \eqref{HE} is surjective.

The family of connections determined by the data \eqref{loop-adhm-swap} may be written explicitly in the 't Hooft ansatz describing $N$ instantons with the same scales and orientations, with positions given by the diagonal components of $M$. Since the positions were chosen inside $\R^3\cong\im(\H)$, the corresponding Skyrme field has precisely the same physical interpretation. In this way, the family \eqref{loop-adhm-swap} generates a loop in the configuration space of Skyrme fields which swaps the positions of two individual skyrmions. This loop is known to be a generator of $\pi_1(S_N)=\Z_2$ \cite{FinkelsteinRubinstein1968}. So both of the homomorphisms $\pi_1(I_N)\to\Z_2$ are surjective, and they therefore must agree as the only surjective homomorphism $\Z_2\to\Z_2$ is the identity.
\end{proof}
\section{Rational maps}
The rational map approximation \cite{HoughtonMantonSutcliffe1998rational} is a popular method of writing down Skyrme fields.  A rational map $R:\mathbb{CP}^1\to\mathbb{CP}^1$ of degree $N$ can be written
\begin{equation}
R(z)=\frac{p(z)}{q(z)},
\end{equation}
where $p=\sum_{i=0}^N p_iz^i$ and $q(z)=\sum_{i=0}^N q_iz^i$ are two complex polynomials satisfying 
\begin{equation}\label{resultant}
\det\begin{pmatrix}p_0&\cdots&\cdots&p_N&&&\\
    &p_0&\cdots&\cdots&p_N&&&\\
    &&\ddots&\cdots&\cdots&\cdots&&\\
    &&&p_0&\cdots&\cdots&\cdots&p_N\\
    q_0&\cdots&\cdots&q_{N-1}&q_N&&&\\
    &q_0&\cdots&\cdots&q_{N-1}&q_N&&\\
    &&\ddots&\cdots&\cdots&\ddots&\ddots&\\
    &&&q_0&\cdots&\cdots&q_{N-1}&\,q_N
    \end{pmatrix}
=:\mathrm{Res}(p,q)\neq 0.
\end{equation}
The Skyrme field $U$ associated to a rational map $R$ takes the form
\begin{equation}
\label{ratmapapprox}
U(\mathbf{x}) = \exp(i\,f(r)\,n^j(z)\sigma_j).
\end{equation}
Here $\sigma_j$ are the Pauli matrices, $r=|\mathbf{x}|$ and $z\in\mathbb{C}$ is a stereographic coordinate defined by $z=(x^1+i x^2)/(r+x^3)$.  The $S^2$-valued function $\mathbf{n}$ is constructed from $R$ using inverse stereographic projection:
\begin{equation}
\mathbf{n}(z) = \left( \frac{2\Re(R(z))}{1+|R(z)|^2},\,\frac{2\Im(R(z))}{1+|R(z)|^2},\,\frac{1-|R(z)|^2}{1+|R(z)|^2}\right).
\end{equation}
Finally, the profile function $f:[0,\infty]\to[0,\pi]$ is required to satisfy $f(0)=\pi$, $f(\infty)=0$ and in practice is usually chosen to minimise an energy.  Suitable choices of rational maps $R$ lead to good approximations to minima of the Skyrme energy \cite{HoughtonMantonSutcliffe1998rational}.

Now suppose that we have a loop $R^t(z)=p^t(z)/q^t(z)$ in the space of rational maps, parametrised by $0\leq t\leq 1$ and satisfying $p^1(z)=p^0(z)$ and $q^1(z)=q^0(z)$.  This induces a loop $U^t(\mathbf{x})$ in the space of Skyrme fields.  The problem of computing the Finkelstein--Rubinstein sign $\sigma(U^t)$ of such a loop was first considered in \cite{krusch2003homotopy}.  Here we present a new method to compute this sign:
\begin{theorem}
Let $p^t,q^t,R^t,U^t$ be as above.  Let $\gamma$ be the loop in $\mathbb{C}^\ast=\mathbb{C}\setminus\{0\}$ given by $\gamma(t)=\mathrm{Res}(p^t,q^t)$.  Let $w(\gamma)\in\mathbb{Z}$ be the winding number of $\gamma$ in $\pi_1(\mathbb{C}^\ast)\cong\mathbb{Z}$.  Then the Finkelstein--Rubinstein sign of the loop $U^t$ of Skyrme fields is given by
\begin{equation}\label{rat map sign}
\sigma(U^t) = (-1)^{w(\gamma)}.
\end{equation}
\end{theorem}
\begin{proof}
Denote by $M_N$ the space of rational maps $R$ of degree $N$, and by $M_N^\ast$ the space of based rational maps satisfying $R(\infty)=\infty$.  It is known that $\pi_1(M_N^\ast)\cong\mathbb{Z}$ and that $\pi_1(M_N)\cong\mathbb{Z}_{2N}$, and that the map $\pi_1(M_N^\ast)\to\pi_1(M_N)$ induced by the inclusion $M_N^\ast\to M_N$ is given by $n\mapsto n\mod 2N$ \cite{Epshtein1973,Segal1979}.

The rational map approximation is a map $M_N\to S_N$.  Krusch has shown that the composition $\pi_1(M_N^\ast)\to\pi_1(M_N)\to\pi_1(S_N)$ induced by the rational map construction is given by $n\mapsto (-1)^n$.  So the map $\pi_1(M_N)\to\pi_1(S_N)\cong \mathbb{Z}_2$ must also be given by $n\mapsto (-1)^n$, where $0\leq n<2N$ represents an element of $\mathbb{Z}_{2N}$.

Now let $P_N$ denote the space of pairs $(p,q)$ of polynomials with $\mathrm{Res}(p,q)\neq0$ and $\max\{\deg(p),\deg(q)\}=N$.  There are natural maps $f_\ast:\pi_1(P_N)\to\pi_1(M_N)\cong\mathbb{Z}_{2N}$ and $\mathrm{Res}_\ast:\pi_1(P_N)\to\pi_1(\mathbb{C}^\ast)\cong\mathbb{Z}$ induced by the projection $f:(p,q)\mapsto R=p/q$ and the resultant $\mathrm{Res}:P_N\to\mathbb{C}^\ast$.  Our preceding comments show that the left hand side of the identity \eqref{rat map sign} is equal to $(-1)^{f_\ast(p^t,q^t)}$, while the right hand side is $(-1)^{\mathrm{Res}_\ast(p^t,q^t)}$, with $(p^t,q^t)$ being a representative loop in $P_N$ of an element of $\pi_1(P_N)$.  To prove the theorem it suffices to show that these two are equal, in other words that the diagram
\begin{equation}\label{cd1}
\begin{tikzcd}
\pi_1(P_N) \arrow[r,"f_\ast"] \arrow[d,"\mathrm{Res}_\ast"] & 
\pi_1(M_N)\cong\mathbb{Z}_{2N} \arrow[d] \\
\pi_1(\mathbb{C}^\ast)\cong\mathbb{Z} \arrow[r] & 
\mathbb{Z}_2
\end{tikzcd}
\end{equation}
commutes.

Now let $P_N^1\subset P_N$ denote the subset of pairs of polynomials with $\mathrm{Res}(p,q)=1$, and consider the following commuting diagram of fibrations.
\begin{equation}\label{fibrations}
\begin{tikzcd}
\mathbb{Z}_{2N}\arrow[r]\arrow[d] & 
\mathbb{C}^\ast\times P^1_N\arrow[r]\arrow[d] & 
P_N \arrow[d,"\mathrm{Res}\times f"] \\
\mathbb{Z}_{2N}\times\mathbb{Z}_{2N}\arrow[r] &
\mathbb{C}^\ast\times P^1_N\arrow[r] & 
\mathbb{C}^\ast\times M_N
\end{tikzcd}
\end{equation}
The map $\mathbb{C}^\ast\times P^1_N\to P_N$ is given by $(a,(p,q))\mapsto (ap,aq)$; this is clearly a $\mathbb{Z}_{2N}$-fibration with fibres consisting of $(e^{-in\pi/N} a,(e^{in\pi/N} p,e^{in\pi/N} q))$ for $n\in\mathbb{Z}_{2N}$.  The central vertical map is the identity.  The map $\mathbb{C}^\ast\times P^1_N\to \mathbb{C}^\ast \times M_N$ is given by $(a,(p,q))\mapsto (a^{2N},p/q)$, which is clearly a $\mathbb{Z}_{2N}\times\mathbb{Z}_{2N}$ fibration with fibres $(e^{im\pi/N} a,(e^{in\pi/N} p,e^{in\pi/N} q))$.  Then the right-most square commutes because $\mathrm{Res}(ap,aq)=a^{2N}\mathrm{Res}(p,q)$.  The left-most vertical map is clearly $n\mapsto (-n,n)$.

From the homotopy exact sequences of these two fibrations we obtain the following commuting diagram.
\begin{equation}
\begin{tikzcd}
\pi_1(P_N) \arrow[r]\arrow[d,"\mathrm{Res}_\ast\times f_\ast"] & 
\pi_0(\mathbb{Z}_{2N})  \arrow[d] \\
\pi_1(\mathbb{C}^\ast)\times \pi_1(M_N) \arrow[r]&
\pi_0(\mathbb{Z}_{2N})\times\pi_0(\mathbb{Z}_{2N})
\end{tikzcd}
\end{equation}
The lower horizontal arrow is a product of two maps.  The first comes from the fibration $\mathbb{Z}_{2N}\to \mathbb{C}^\ast\to\mathbb{C}^\ast$ and is given by $\pi_1(\mathbb{C}^\ast)\cong\mathbb{Z}\to \mathbb{Z}_{2N}$, $n\mapsto n\mod 2N$.  The second is the map $\mathbb{Z}_{2N}\cong\pi_1(M_N)\to \pi_0(\mathbb{Z}_{2N})$ which arises from the fibration $\mathbb{Z}_{2N}\to P_N^1\to M_N$; this is known to be an isomorphism \cite{Epshtein1973}.  The image of the right-most vertical arrow consists of $(-n,n)$ for $n\in\mathbb{Z}_{2N}$.  Since the diagram commutes, it follows that
\begin{equation}
f_\ast(p^t,q^t)= -\mathrm{Res}_\ast(p^t,q^t) \mod{2N}
\end{equation}
for all loops $(p^t,q^t)$ in $P_N$.  This equality modulo $2N$ implies equality modulo $2$, and hence the diagram \eqref{cd1} commutes as claimed.
\end{proof}
\section{Examples}
In this section we demonstrate the methods derived above in some examples.

Numerical studies have shown that the minimal-energy $N=3$ Skyrme field has tetrahedral symmetry.  The quantisation of the 3-skyrmion presented in \cite{walet1996quantising} used a path in the $N=3$ configuration space of Skyrme fields which starts at this tetrahedron, passes through a torus, and ends at the dual tetrahedron. The dual tetrahedron may then be rotated and isorotated to go back to the original tetrahedron, forming a loop in the configuration space $S_3$. Here we construct this path using both ADHM data and rational maps, and compute the corresponding Finkelstein--Rubinstein sign.

Consider the following one-parameter family of ADHM data
\begin{align}\label{D2d-data}
\begin{aligned}
    L(t)&=\begin{pmatrix}
        \bi&\bj&t\bk
    \end{pmatrix},&
    M(t)&=\begin{pmatrix}
    0&t\bk&\bj\\
    t\bk&0&\bi\\
    \bj&\bi&0\end{pmatrix}.
\end{aligned}
\end{align}
The data has $D_{2d}$ symmetry for all $t\in\R$. At $t=\pm1$ the data has tetrahedral symmetry $T_d$, and at $t=0$ the data has toroidal symmetry $D_{\infty h}$.  So varying $t$ between $-1$ and $1$ and applying the Atiyah--Manton construction reproduces the path of Skyrme fields found in \cite{walet1996quantising}. Charge density isosurface plots of this path at $t=-1,-0.5,0,0.5,1$ are given in Figure \ref{fig:D2d-path}.
\begin{figure}[ht]
    \centering
    \includegraphics[scale=0.5,keepaspectratio=true]{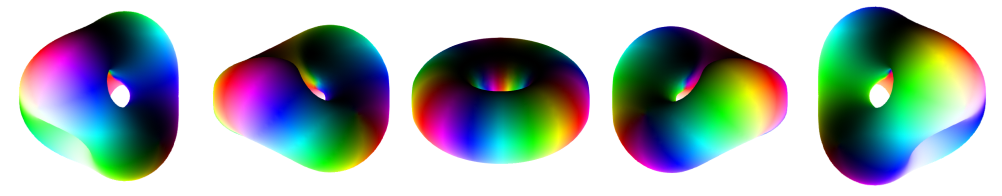}
    \caption{The path in the $N=3$ configuration space generated by the data \eqref{D2d-data}.}
    \label{fig:D2d-path}
\end{figure}

The data at $t=-1$ and $t=1$ are related as follows:
\begin{align*}
    \begin{aligned}
        L(1)&=qL(-1)q^{-1}P^{-1},&M(1)&=qP M(-1)P^{-1}q^{-1},
    \end{aligned}
\end{align*}
where $q=\tfrac{\sqrt{2}}{2}(\bu+\bk)$ and $P=\begin{pmatrix}
    0&-1&0\\
    1&0&0\\
    0&0&-1
\end{pmatrix}$. In Skyrme language this means the tetrahedra are related by a rotation of space about the $x_3$-axis through angle $\tfrac{\pi}{2}$ and an isorotation of the target 3-sphere, also through angle $\tfrac{\pi}{2}$. Thus, the path formed by combining the path from $t=-1$ to $t=1$ with this rotation-isorotation pair is a loop in configuration space. Since $\det P=-1$, the Finkelstein--Rubinstein sign is $-1$. In other words, the loop is non-contractible.

We can consider the same path of $N=3$ Skyrme fields using rational maps.  The family corresponding to \eqref{D2d-data} is
\begin{equation}
R(z)=\frac{\sqrt{3}itz^2-1}{z^3-\sqrt{3}tiz}
\end{equation}
These satisfy $R(-z)=-R(z)$, $R(1/z)=1/R(z)$, and $R(i\bar{z})=i\overline{R(z)}$, so they have $D_{2d}$ symmetry.  The maps with $t=\pm1$ have tetrahedral $T_d$ symmetry, and the maps with $t=0$ have axial symmetry about the $x_3$-axis.  The tetrahedral map with $t=-1$ is related to the map with $t=1$ by a transformation
\begin{equation}
R(z)\mapsto -i R(iz).
\end{equation}
As before this is a combined rotation and isorotation through angles of $\frac{\pi}{2}$ about the $x_3$-axis.

The path obtained by varying $t$ from $1$ to $-1$ and then rotating and isorotating through angle $\tfrac{\pi}{2}$ is a closed loop in the space of rational maps.  To calculate its Finkelstein--Rubinstein sign, we calculate the resultant of the map
\begin{equation}
e^{-i\theta}R(e^{i\theta}z)=e^{-4i\theta}\frac{e^{2i\theta}\sqrt{3}itz^2-1}{z^3-e^{-2i\theta}\sqrt{3}itz}.
\end{equation}
The resultant turns out to be
\begin{equation}
e^{-12i\theta}\begin{vmatrix}
-1 & 0 & e^{2i\theta}\sqrt{3}it & 0 & 0 & 0 \\
0 & -1 & 0 & e^{2i\theta}\sqrt{3}it & 0 & 0 \\
0 & 0 & -1 & 0 & e^{2i\theta}\sqrt{3}it & 0 \\
0 & -e^{-2i\theta}\sqrt{3}it & 0 & 1 & 0 & 0 \\
0 & 0 & -e^{-2i\theta}\sqrt{3}it & 0 & 1 & 0 \\
0 & 0 & 0 & -e^{-2i\theta}\sqrt{3}it & 0 & 1 \\
\end{vmatrix}
=-e^{-12i\theta}(1+3t^2)^2.
\end{equation}
As $t$ varies from $1$ to $-1$ and $\theta$ varies from $0$ to $\frac{\pi}{2}$ this circles the origin three times.  So the winding number is $3$ and the Finkelstein--Rubinstein sign is once again $(-1)^3=-1$.

As a final example, we consider paths induced by symmetries of rational maps, which were studied earlier in \cite{krusch2003homotopy}.  Suppose that a rational map $R$ is invariant under a simultaneous rotation of its domain through angle $\alpha$ and isorotation of its target through angle $\beta$.  Without loss of generality we may assume that both rotations are about the $x_3$-axis and hence that
\begin{equation}
R(z)=e^{i\beta}R(e^{-i\alpha}z).
\end{equation}
The left and right hand sides of this equation are the start and end points of the following path in the space of rational maps:
\begin{equation}
\begin{aligned}
p^t(z) &= e^{it\beta}(p_0+p_1e^{-it\alpha}z+\ldots+p_Ne^{-iNt\alpha}z^N)\\
q^t(z) &= (q_0+q_1e^{-it\alpha}z+\ldots+q_Ne^{-iNt\alpha}z^N).
\end{aligned}
\end{equation}
Recall that the resultant of $p$ and $q$ is the determinant of a matrix $T(p,q)$ given in \eqref{resultant}.  We find that
\begin{align}
M(p^t,q^t)&=D_1 T(p^0,q^0)D_2^{-1}\\
D_1&=\mathrm{diag}(e^{it(\alpha+\beta)},e^{it(2\alpha+\beta)},e^{it(N\alpha+\beta)},e^{it\alpha},e^{2it\alpha},\ldots,e^{Nit\alpha})\\
D_2&=\mathrm{diag}(e^{it\alpha},e^{2it\alpha},\ldots,e^{2Nit\alpha}).
\end{align}
Hence
\begin{equation}
\mathrm{Res}(p^t,q^t)=\det(D_1)\det(D_2^{-1})\mathrm{Res}(p^0,q^0)=e^{itN(\beta-N\alpha)}\mathrm{Res}(p^0,q^0).
\end{equation}
The Finkelstein--Rubinstein sign is calculated from the winding number of this path in $\mathbb{C}^\ast$ and is given by
\begin{equation}
(-1)^{N(\beta-N\alpha)/2\pi},
\end{equation}
in agreement with \cite{krusch2003homotopy}.
\section{Conclusion and outlook}
We have established two new methods for computing Finkelstein--Rubinstein signs within the widely-used constructions of Skyrme fields: instantons and rational maps. Both methods are simple to use and can be applied to any loop; in particular, their use is not restricted to loops generated by symmetries.  For symmetry-generated loops within the rational map approximation, we have obtained a simple derivation of the method introduced in \cite{krusch2003homotopy} as a special case of our more general method. The methods developed here will prove useful in ongoing efforts to quantise skyrmions using paths of Skyrme fields \cite{corkhalcrow2022adhm,SpeightWinyard2023nudged}.
\subsection*{Acknowledgements} We thank Chris Halcrow for useful suggestions and encouragement. Figure \ref{fig:D2d-path} was generated using the Skyrmions3D software developed by Chris Halcrow.
\bibliographystyle{plain}
\bibliography{refs}

\begin{thebibliography}{10}

\bibitem{ADHM1978construction}
M~F Atiyah, V~G Drinfeld, N~J Hitchin, and Y~I Manin.
\newblock Construction of instantons.
\newblock {\em Phys. Lett. A}, 65(3):185--187, 1978.

\bibitem{AtiyahManton1989}
M~F Atiyah and N~S Manton.
\newblock Skyrmions from instantons.
\newblock {\em Phys. Lett. B}, 222(3):438--442, 1989.

\bibitem{BattyeMantonSutcliffeWood2009light}
R~A Battye, N~S Manton, P~M Sutcliffe, and S~W Wood.
\newblock {Light nuclei of even mass number in the Skyrme model}.
\newblock {\em Phys. Rev. C}, 80(3):034323, 2009.

\bibitem{christWeinbergStanton1978general}
N~H Christ, E~J Weinberg, and N~K Stanton.
\newblock General self-dual {Yang-Mills} solutions.
\newblock {\em Phys. Rev. D}, 18(6):2013, 1978.

\bibitem{corkhalcrow2022adhm}
J~Cork and C~Halcrow.
\newblock {ADHM} skyrmions.
\newblock {\em Nonlinearity}, 35(8):3944, 2022.

\bibitem{Epshtein1973}
S~I \'{E}pshtein.
\newblock Fundamental groups of spaces of coprime polynomials.
\newblock {\em Funct Anal Its Appl}, 7:82--83, 1973.

\bibitem{FinkelsteinRubinstein1968}
D~Finkelstein and J~Rubinstein.
\newblock Connection between spin, statistics, and kinks.
\newblock {\em J. Math. Phys.}, 9(11):1762--1779, 1968.

\bibitem{GillardHarlandKirkElliotMaybeeSpeight2017}
M~Gillard, D~Harland, E~Kirk, B~Maybee, and M~Speight.
\newblock {A point particle model of lightly bound skyrmions}.
\newblock {\em Nucl. Phys. B}, 917:286--316, 2017.

\bibitem{Gudnason:2018aej}
S~B Gudnason and C~Halcrow.
\newblock {$B=5$ Skyrmion as a two-cluster system}.
\newblock {\em Phys. Rev. D}, 97(12):125004, 2018.

\bibitem{Halcrow:2022bxw}
C~Halcrow and D~Harland.
\newblock {Nucleon-nucleon potential from instanton holonomies}.
\newblock {\em Phys. Rev. D}, 106(9):094011, 2022.

\bibitem{Halcrow:2015rvz}
C~J Halcrow.
\newblock {Vibrational quantisation of the B = 7 Skyrmion}.
\newblock {\em Nucl. Phys. B}, 904:106--123, 2016.

\bibitem{Halcrow:2016spb}
C~J Halcrow, C~King, and N~S Manton.
\newblock {A dynamical $\alpha$-cluster model of $^{16}$O}.
\newblock {\em Phys. Rev. C}, 95(3):031303, 2017.

\bibitem{HoughtonMantonSutcliffe1998rational}
C~J Houghton, N~S Manton, and P~M Sutcliffe.
\newblock Rational maps, monopoles and skyrmions.
\newblock {\em Nucl. Phys. B}, 510(3):507--537, 1998.

\bibitem{hurtubise1986instantons}
J~Hurtubise.
\newblock Instantons and jumping lines.
\newblock {\em Commun. Math. Phys.}, 105(1):107--122, 1986.

\bibitem{krusch2003homotopy}
S~Krusch.
\newblock Homotopy of rational maps and the quantization of skyrmions.
\newblock {\em Ann. Phys.}, 304(2):103--127, 2003.

\bibitem{krusch2006finkelstein}
S~Krusch.
\newblock {Finkelstein--Rubinstein} constraints for the {Skyrme} model with
  pion masses.
\newblock {\em Proc. R. Soc. Lond. A}, 462(2071):2001--2016, 2006.

\bibitem{LauManton2014}
P~H~C Lau and N~S Manton.
\newblock {Quantization of $T_d$- and $O_h$-symmetric Skyrmions}.
\newblock {\em Phys. Rev. D}, 89(12):125012, 2014.

\bibitem{Leese:1994hb}
R~A Leese, N~S Manton, and B~J Schroers.
\newblock {Attractive channel skyrmions and the deuteron}.
\newblock {\em Nucl. Phys. B}, 442:228--267, 1995.

\bibitem{mankomantonwood2007}
O~V Manko, N~S Manton, and S~W Wood.
\newblock Light nuclei as quantized skyrmions.
\newblock {\em Phys. Rev. C}, 76(5):055203, 2007.

\bibitem{MantonSutcliffe2004}
N~S Manton and P~M Sutcliffe.
\newblock {\em Topological solitons}.
\newblock Cambridge University Press, 2004.

\bibitem{MantonWood2006lithium}
N~S Manton and S~W Wood.
\newblock {Reparametrizing the Skyrme model using the lithium-6 nucleus}.
\newblock {\em Phys. Rev. D}, 74(12):125017, 2006.

\bibitem{Rawlinson:2017rcq}
J~I Rawlinson.
\newblock {An alpha particle model for carbon-12}.
\newblock {\em Nucl. Phys. A}, 975:122--135, 2018.

\bibitem{Segal1979}
G~Segal.
\newblock Topology of spaces of rational functions.
\newblock {\em Acta Math.}, 143(1-2):39--72, 1979.

\bibitem{skyrme1962nucl}
T~H~R Skyrme.
\newblock A unified field theory of mesons and baryons.
\newblock {\em Nucl. Phys.}, 31:556, 1962.

\bibitem{SpeightWinyard2023nudged}
J~M Speight and T~Winyard.
\newblock Nudged elastic bands and lightly bound skyrmions.
\newblock {\em SIGMA}, 19(073), 2023.

\bibitem{walet1996quantising}
N~R Walet.
\newblock Quantising the {$B= 2$} and {$B= 3$ Skyrmion} systems.
\newblock {\em Nucl. Phys. A}, 606(3-4):429--458, 1996.

\bibitem{Witten1983baryons}
E~Witten.
\newblock {Current algebra, baryons, and quark confinement}.
\newblock {\em Nucl. Phys. B}, 223:433--444, 1983.

\end{thebibliography}
\end{document}